\documentclass{lmcs} 
\pdfoutput=1

\usepackage{lastpage}
\lmcsdoi{20}{3}{21}
\lmcsheading{}{\pageref{LastPage}}{}{}%
{May~23,~2023}{Sep.~04,~2024}{}

\pdfoutput=1


\usepackage[utf8]{inputenc}
\usepackage[T1]{fontenc}
\usepackage{hyperref}
\usepackage{amsmath}
\usepackage{amsthm}
\usepackage{amsfonts}
\usepackage{amssymb}
\usepackage{stmaryrd}
\usepackage{graphicx}
\usepackage{xspace}
\usepackage{mathtools}
\usepackage{csquotes}

\newcommand{\Str}{\mathtt{Str}}
\newcommand{\Bool}{\mathtt{Bool}}
\newcommand{\Nat}{\mathtt{Nat}}
\newcommand{\naturals}{\mathbb{N}}
\newcommand{\rmord}{\mathrm{ord}}
\newcommand{\rmheight}{\mathrm{height}}
\newcommand{\inford}{\mathrm{inford}}

\newcommand{\tower}{\mathrm{tower}}
\newcommand{\tttow}{\mathtt{tow}}
\newcommand{\tttrue}{\mathtt{true}}
\newcommand{\ttfalse}{\mathtt{false}}
\newcommand{\ttenum}{\mathtt{enum}}
\newcommand{\ttany}{\mathtt{any}}
\newcommand{\ttconc}{\mathtt{cat}}
\newcommand{\ttsplit}{\mathtt{split}}
\newcommand{\ttand}{\mathtt{and}}
\newcommand{\ttor}{\mathtt{or}}
\newcommand{\ttnot}{\mathtt{not}}

\newcommand{\emptycode}{{\overline\varepsilon}}

\newcommand\Elementary{\textsc{Elementary}\xspace}
\newcommand\ExpTime{\textsc{ExpTime}\xspace}

\newcommand\FElementary{\textsc{FElementary}\xspace}
\newcommand\FExpTime{\textsc{FExpTime}\xspace}
\newcommand\PSpace{\textsc{PSpace}\xspace}
\newcommand\Tower{\textsc{Tower}\xspace}

\newcommand\hls{\vdash_{\mathsf{hls}}}
\newcommand\denot[1]{\llbracket #1 \rrbracket}
\newcommand\pared[2][]{#2^{\star #1}}

\begin{document}
\title[Simply typed convertibility is TOWER-complete even for safe $\lambda$-terms]{Simply typed convertibility is TOWER-complete\texorpdfstring{\\}{} even for safe lambda-terms}
\thanks{Thanks to Anupam Das, Damiano Mazza and Noam Zeilberger for many
  instructive discussions on the topic of complexity of normalization for $\lambda$-calculi.}	

\thanks{The author was supported by the LambdaComb project
  (ANR-21-CE48-0017) while working at École Polytechnique and by the LABEX MILYON
  (ANR-10-LABX-0070) of Université de Lyon, within the program
  \enquote{France 2030} operated by the French National Research
  Agency (ANR)}

\author[L.~T.~D.~{Nguy\~{\^e}n}]{{Lê Thành D\~ung (Tito) Nguy\~{\^e}n}\lmcsorcid{0000-0002-6900-5577}}

\address{Laboratoire de l'informatique du parallélisme (LIP), École normale supérieure de Lyon, France}
\email{nltd@nguyentito.eu}





\begin{abstract}
  \noindent We consider the following decision problem: given two simply typed
  $\lambda$-terms, are they $\beta$-convertible? Equivalently, do they have the
  same normal form? It is famously non-elementary, but the precise complexity --
  namely \Tower-complete -- is lesser known. One goal of this short paper is
  to popularize this fact.

  Our original contribution is to show that the problem stays
  \Tower-complete when the two input terms belong to Blum and Ong's safe
  $\lambda$-calculus, a fragment of the simply typed $\lambda$-calculus arising
  from the study of higher-order recursion schemes. Previously, the best known
  lower bound for this safe $\beta$-convertibility problem was
  \PSpace-hardness. Our proof proceeds by reduction from the star-free
  expression equivalence problem, taking inspiration from the author's work with
  Pradic on \enquote{implicit automata in typed $\lambda$-calculi}.

  These results also hold for $\beta\eta$-convertibility.
\end{abstract}

\maketitle

\section{Introduction}

Consider the following \emph{simply typed $\beta$-convertibility} problem:
\begin{itemize}
\item \emph{Input:} two simply typed $\lambda$-terms $t$ and $u$ of the same type -- this
  involves some choices of representation for such terms, that will be
  discussed later (\S\ref{sec:pedantic}).
\item \emph{Output:} does $t =_\beta u$ hold? Equivalently -- since
  $\beta$-reduction is confluent and terminating in the simply typed
  $\lambda$-calculus -- do $t$ and $u$ have the same normal form?
\end{itemize}
This is, of course, a fundamental decision problem concerning the
$\lambda$-calculus. For example, the special case when $t : \mathtt{Bool} = o
\to o \to o$ and $u = \mathtt{true} = \lambda x^o.\,
\lambda y^o.\, x$ amounts to \emph{evaluating} the result of the program $t$. (Here, $o$ is the single base type that we work with: our grammar of simple types is $A,B \mathrel{::=} o \mid A \to B$.)
Also, proof assistants based on dependent types often need to check whether two
terms are \emph{definitionally equal} -- which is nothing more than the
generalization of the above problem to more sophisticated type theories.
For these reasons, there have been many works on deciding convertibility
efficiently, such as~\cite{Condoluci}.

To find out whether $t =_{\beta} u$, a naive approach is to normalize both $t$ and
$u$ and then compare their normal forms. However, it is not always the best
course of action: see, for instance, the practical remarks
of~\cite[Section~4]{Asperti17a} or the use of denotational semantics to decide
convertibility in~\cite{Terui}.

\subsection{Complexity of \texorpdfstring{$\beta$}{beta}-convertibility}

A well-known result of Statman~\cite{Statman} is that simply typed
$\beta$-convertibility is \emph{not} in the complexity class \Elementary, whose
definition we recall now. Let $2_0^n = n$ and $2_{k+1}^n = 2^{2_k^n}$. A
function is in the class $k$-\FExpTime when it can be computed in time bounded
by $2_k^{P(\text{input size})}$ for some polynomial $P$; and the class
\FElementary of \emph{elementary recursive functions} is the union of all the
$k$-\FExpTime classes for $k\in\naturals$. By \Elementary, we refer to the class of
decision problems that, when seen as functions returning a boolean, belong to
\FElementary.

The statement of Statman's theorem can be strengthened as follows. Let
$\tower(n) = 2_n^1$. The complexity class \Tower consists of the decision
problems that can be solved in time $\tower(2_k^{P(\text{input size})})$ for
some polynomial $P$ and some $k\in\naturals$. A problem is \Tower-hard when every
problem in \Tower reduces to it via many-one \FElementary reductions. By the
time hierarchy theorem, a \Tower-hard problem cannot be in \Elementary.
\begin{thm}\label{thm:stlc-intro}
  Simply typed $\beta$-convertibility is \Tower-complete (that is, it belongs to
  the class \Tower, and is at the same time \Tower-hard).
\end{thm}
The hardness part is actually implicit in Statman's proof~\cite{Statman}: it
provides a reduction from a problem which is essentially the same as the
\Tower-complete Higher-Order Quantified Boolean Satisfiability (HOSAT) problem.
For a more detailed discussion of this point, see the \enquote{related work}
paragraph of~\cite[Section~3]{HOSAT}; let us also mention that
Mairson~\cite{MairsonSimple} has given a simpler reduction from HOSAT to simply typed
$\beta$-convertibility. In fact, in his relatively recent
article~\cite{Schmitz} where the complexity class \Tower is explictly introduced
for the first time, Schmitz observes that many non-elementary lower bounds in
the literature can be read as proofs of \Tower-hardness -- and he mentions the
example of simply typed $\beta$-convertibility~\cite[\S3.1.2]{Schmitz}.

As for membership in \Tower, it is known to be the \enquote{easy part} of Theorem~\ref{thm:stlc-intro}. In Section~\ref{sec:stlc}, we recall some folklore ideas about normalization in the simply typed $\lambda$-calculus, and observe that they yield a \Tower algorithm for $\beta$-convertibility. Section~\ref{sec:stlc} also discusses how the complexity of normalizing a $\lambda$-term is controlled by the \emph{degree} and the \emph{order} of the types that it contains -- defined respectively as $\deg(o) = \rmord(o) = 0$ and
\[  \deg(A\to B) = \max(\deg(A),\deg(B))+1 \qquad
  \rmord(A\to B) = \max(\rmord(A)+1,\, \rmord(B)) \]
In other words, the degree of a type is the height of its syntax tree, while the order is the nesting depth of function arrows to the left.

\subsection{The case of the safe \texorpdfstring{$\lambda$}{lambda}-calculus}\label{sec:intro-safe}

Coincidentally, the order of a simple type is also central to the definition of
\emph{safe} fragment of the simply typed $\lambda$-calculus, even though the
motivation that led Blum and Ong~\cite{BlumOng,BlumPhD} to introduce this
fragment were unrelated to computational complexity. (We discuss this background
in Section~\ref{sec:hors}.)

Let us recall briefly the definition of the safe $\lambda$-calculus. Let $t$ be a simply typed $\lambda$-term in \enquote{Church style} -- i.e., the variables have type annotations ensuring that the type of each subterm of $t$ is uniquely determined. Then $t$ is \emph{unsafe} if it contains a subterm $u : A$ that
\begin{itemize}
  \item contains some free variable $x : B$ with $\rmord(B) < \rmord(A)$
  \item while being in \enquote{unsafe position}, that is:
  \begin{itemize}
    \item either $u$ equals $t$ itself,
    \item or $u$ is in argument position: $t = C[v\;u]$ for some term $v$ and one-hole context $C[\cdot]$,
    \item or $u$ is applied to some argument ($t = C[u\;v]$) but $u$ itself is not an application.
  \end{itemize}
\end{itemize}
A safe $\lambda$-term is, obviously, a simply typed $\lambda$-term that is not unsafe. For example:
\[ \lambda f^{(o\to o)\to o}.\; f\;(\lambda x^o.\; f\;(\underbrace{\lambda y^o.\; \overset{\underset{\downarrow}{\mathclap{\text{free variable of type}\ o ~~\rightsquigarrow~~ \text{order 0}}}}{x}}_{\mathclap{\text{argument subterm of type}\ (o \to o)~~\rightsquigarrow~~ \text{order 1}}}))\ \text{is unsafe} \qquad\quad \lambda f^{(o\to o)\to o}.\; f\;(\lambda x^o.\; f\;(\lambda y^o.\; y))\ \text{is safe} \]
The precise complexity of deciding $\beta$-convertibility on safe $\lambda$-terms has been an open problem until now. In their paper introducing the safe $\lambda$-calculus, Blum and Ong prove that it is \PSpace-hard~\cite[Section~3]{BlumOng}. This lower bound is far below \Tower, but they argue that both Statman's proof~\cite{Statman} of the hardness part in Theorem~\ref{thm:stlc-intro} and its simplification by Mairson~\cite{MairsonSimple} fundamentally require unsafe terms. To quote~\cite[\S3.1]{BlumOng}, this \enquote{does not rule out the possibility that another non-elementary problem is encodable in the safe lambda calculus}. This is precisely what our original contribution here is about:
\begin{thm}\label{thm:safe-intro}
  $\beta$-convertibility in the safe $\lambda$-calculus is \Tower-complete.
\end{thm}
Membership in \Tower follows, of course, from the \enquote{easy part} of Theorem~\ref{thm:stlc-intro}, which is the topic of Section~\ref{sec:stlc}. We establish \Tower-hardness in Section~\ref{sec:safe} even for the special case of \enquote{homogeneous long-safe} $\lambda$-terms of type $\Bool$ (Theorem~\ref{thm:safe-hard}), by reduction from a standard problem in automata theory: \emph{star-free expression equivalence}. As a corollary, this provides an alternative proof of the \enquote{hard part} of Theorem~\ref{thm:stlc-intro}.

\subsection{Some pedantic points}
\label{sec:pedantic}

Traditionally, the simply typed $\lambda$-calculus (and its safe fragment) can be presented in two ways (cf.~\cite{ChurchvsCurry}). We already mentioned the intrinsically typed \enquote{Church style}. By contrast, in \enquote{Curry style}, a $\lambda$-term is given without type annotations, and may satisfy several different typing judgments (e.g.\ $\vdash \lambda xy.\,x : A \to B \to A$ for all simple types $A$ and $B$). We might therefore consider two extreme choices when specifying the input for the simply typed $\beta$-convertibility problem, that might \textit{a priori} make a difference regarding computational complexity:
\begin{enumerate}
  \item\label{item:church} All variables in the input terms are fully annotated with their types, and these annotations are counted in the size of the input.
  \item\label{item:curry} There are no types in the input at all: we are given two untyped terms $t$ and $u$, with the promise that there exists some simple type $A$ such that $\vdash t : A$ and $\vdash u : A$.
\end{enumerate}
The mathematical equivalence of these two questions already requires some care. It is not the case in general for type theories that type erasure is injective modulo conversion: for example, in the polymorphic $\lambda$-calculus (System~F), there are Church-style $\lambda$-terms of the same type that are not convertible even though the underlying untyped terms are $\beta$-convertible.\footnote{In System F, the latter condition is called \enquote{Strachey equivalence}~\cite[\S2.2]{PlotkinAbadi}, see also~\cite{TraceF}.} However, this does not happen in the simply typed $\lambda$-calculus, because:
\begin{fact}\label{fact:injective-erasure}
  A Church-style simply typed $\lambda$-term \emph{in normal form}\footnote{The normal form assumption is required to exclude counter-examples such as $(\lambda xy.\, y)\, (\lambda z.\, z) : A \to A$ where, even when a specific $A$ is fixed, $x$ can be given any type of the form $B \to B$.} can be reconstructed uniquely from its type and its underlying untyped term.
\end{fact}
Complexity-wise, (\ref{item:church}) is easier than (\ref{item:curry}): simply erasing all types in the input gives a trivial reduction. Conversely, type inference for the simply typed $\lambda$-calculus can be performed in linear time, using first-order unification (see e.g.~\cite{WandInference}) -- this also works for inferring a common simple type for a pair $(t,u)$ of $\lambda$-terms, by adding an equation to the unification problem. A solution to this unification problem -- of size $O(|t|+|u|)$ -- also gives us the type annotations to add to $t$ and $u$; we can ask a unification algorithm to return a unifier that has linear size using a representation of the syntax trees as directed acyclic graphs (in order to share subtrees), cf.~\cite[Sections~4.8~and~4.9]{TRAAT}. Note that this guarantees that $\rmord(A) \leqslant \deg(A) = O(|t|+|u|)$ for any type $A$ among the computed annotations.

Thus, there is a linear time reduction from (\ref{item:curry}) to the variant (\ref{item:church}') of (\ref{item:church}) where such shared representations are allowed for the types. In turn, (\ref{item:church}') reduces to (\ref{item:church}) in exponential time by unfolding the syntax trees of the types. Since $\text{1-\FExpTime} \subset \text{\FElementary}$, this does not make a difference regarding membership in or hardness for \Tower.

\subsection{The \texorpdfstring{$\eta$}{eta} rule}

Theorems~\ref{thm:stlc-intro} and~\ref{thm:safe-intro} also hold for $\beta\eta$-conversion.

For the hardness part, this is because we prove it for the problem restricted to $\lambda$-terms of type $\Bool$, and $\beta\eta$-convertibility coincides with $\beta$-convertibility at type $\Bool$. (Indeed, for every $t : \Bool$, either $t =_\beta \tttrue$ or $t =_\beta \ttfalse = \lambda x^o.\,\lambda y^o.\,y$, and $\tttrue \neq_{\beta\eta} \ttfalse$.)

For the complexity upper bound, note that Section~\ref{sec:stlc} actually gives an algorithm for computing the $\beta$-normal form $t'$ of a simply typed $\lambda$-term $t$. Once we have $t'$, applying \mbox{$\eta$-reductions} takes time $|t'|^{O(1)}$ until a $\beta\eta$-normal form\footnote{Also known as a $\beta$-normal $\eta$-\emph{short} form. In many other contexts, $\beta$-normal $\eta$-\emph{long} forms are more useful, but here $\eta$-reduction is more convenient than $\eta$-expansion because it gives us a terminating rewriting system.} is reached. And testing whether two terms are \mbox{$\beta\eta$-convertible} can be done by comparing their \mbox{$\beta\eta$-normal} forms for equality\footnote{In this paper, we always consider $\lambda$-terms modulo $\alpha$-renaming, so an equality test that takes concrete syntax trees as input has to take $\alpha$-renaming into account.}, thanks to the confluence of $\beta\eta$-reduction on simply typed $\lambda$-terms of a fixed type -- a fact shown in~\cite{Geuvers} which is a bit more subtle than one could expect. This also means that we can invoke Fact~\ref{fact:injective-erasure} again to conclude that the Church-style and Curry-style variants of the $\beta\eta$-convertibility problem are equivalent.

But this does not work anymore when we add sum types endowed with the
\enquote{strong} $\eta$ rule: in this setting, \enquote{there can be no
  $\lambda$-calculus with sums \textit{à la}
  Curry}~\cite[\S{}I(a)]{gascheIntermediate} because untyped
$\beta\eta$-conversion is inconsistent (that is, it equates all
$\lambda$-terms). In fact, decidability of $\beta\eta$-convertibility for a
simply typed $\lambda$-calculus with sums and the empty type, using typed
conversion rules, has only been proved relatively recently~\cite{gascheEmpty}. As future work, it could be interesting to know whether this problem is still in \Tower.

\section{Simply typed \texorpdfstring{$\beta$}{beta}-convertibility is in Tower}
\label{sec:stlc}

\subsection{A simple \Tower normalization algorithm using the degree}

Let us first define the \emph{type of a redex} $\big(\lambda x^A.\, t\big)\, u$ in a Church-style simply typed $\lambda$-term as $A \to B$ where $B$ is the type of $t$ (so $A\to B$ is also the type of $\lambda x^A.\, t$). The \emph{degree of a redex} is the degree of its type.
The idea of using redex degrees to get a weak normalization proof for the simply typed $\lambda$-calculus goes back at least to Turing, see~\cite[Section~1]{BarenbaumSottile} for an explanation and historical discussion. Here, we use a variation on this idea based on the notion of \emph{parallel reduction}, which comes from the Tait/Martin-Löf proof of confluence for untyped $\beta$-reduction.

\begin{defiC}[{\cite[p.~121]{Takahashi}}]\label{def:pared}
  The parallel reduction $\pared{t}$ of a $\lambda$-term $t$ is, inductively:
  \[ \pared{x} = x \qquad \pared{(\lambda x.\, t)} = \lambda x.\, \pared{t} \qquad
    \pared{(t\,u)} = \begin{cases}
      \pared{(t')}\{x:=\pared{u}\} &\text{when}\ t = \lambda x.\,t'\ \text{for some}\ t'\\
      \pared{t}\,\pared{u} & \text{otherwise}
    \end{cases} 
  \]
  This definition on untyped $\lambda$-terms can be extended to Church-style simply typed $\lambda$-terms in the obvious way.
\end{defiC}
It is immediate that $t \longrightarrow_\beta^* \pared{t}$. Furthermore:
\begin{propC}[{\cite[Appendix]{AspertiLevy}}]
  Let $t$ be a simply typed $\lambda$-term whose redexes have degree at most $d\in\naturals$. Then:
  \begin{itemize}
    \item The redexes of $\pared{t}$ have degree at most $d-1$.
    \item As a consequence, the normal form of $t$ can be reached in $d$ parallel reductions. In other words, let $\pared[0]{t} = t$ and $\pared[(n+1)]{t} = \pared{(\pared[n]{t})}$; then $\pared[d]{t}$ is the normal form of $t$.
  \end{itemize}
\end{propC}
This gives us a simple algorithm to compute the normal form of a simply typed $\lambda$-term. In fact, since parallel reduction does not inspect types, we may also apply it to simply \emph{typable} terms, to solve the Curry-style version of the $\beta$-convertibility problem (cf.~\S\ref{sec:pedantic}).

We consider the following measures on an untyped $\lambda$-term $t$ to get a complexity bound:
\begin{itemize}
  \item $\rmheight(t)$ denotes the height of its syntax tree;
  \item its \emph{size} $|t|$ is the number of nodes of its syntax tree, including the leaves (variables): \[ |x| = 1 \qquad |\lambda x.\, t| = 1+|t| \qquad |t\,u| = 1+|t|+|u|\]
\end{itemize}
\begin{fact}
  For any untyped $\lambda$-term $t$, we have $\rmheight(\pared{t}) \leqslant |t| \leqslant 2^{\rmheight(t)}$.
\end{fact}
(We have not managed to find this statement in the literature, but it is probably not new.)
\begin{proof}
  The first inequality can be proved by structural induction on Definition~\ref{def:pared}, using
  \[ \rmheight(\pared{(t')}\{x:=\pared{u}\}) \leqslant \rmheight(\pared{(t')})+\rmheight(\pared{u})\]
  while the second inequality is due to the fact that each node in the syntax tree of a $\lambda$-term has at most two children.
\end{proof}

As a consequence, $|\pared{t}| \leqslant 2^{|t|}$ for any $\lambda$-term $t$. Putting all this together, we get:
\begin{thm}
  There is an algorithm that takes as input a simply typable $\lambda$-term $t$ and returns its normal form in time at most $2_d^{P(|t|)}$, where $d$ is the maximum redex degree in some typing of $t$ (not given as input) and $P$ is some polynomial independent from $t$ and $d$.
\end{thm}
\begin{proof}
  The algorithm iterates parallel reductions until it reaches a normal form. This takes at most $d$ steps; for each step $n\in\{1,\dots,d\}$, computing $\pared[n]{t}$ can be done by feeding the input $\pared[(n-1)]{t}$ -- which has size at most $2_{n-1}^{|t|}$ -- to a naive 1-\FExpTime procedure.
\end{proof}
Given two simply typable $\lambda$-terms, we can compute their normal forms and compare them to decide whether they are $\beta$-convertible. Together with the considerations from \S\ref{sec:pedantic}, this establishes the \enquote{membership in \Tower{}} part of Theorem~\ref{thm:stlc-intro}. (Again, we do not make any claim to originality concerning the material in this section; it has been included for expository purposes.)

\subsection{Finer bounds with the order}\label{sec:stlc-order}

Actually, a \Tower algorithm for normalization could also be obtained just by
applying successive (non-parallel) $\beta$-reductions. The complexity then
depends on the length of reduction sequences for simply typed $\lambda$-terms,
for which bounds are known~\cite{Schwichtenberg,Beckmann}. But note that these bounds are stated using the
\emph{order} of the types in a $\lambda$-term, not their degree!
\begin{rem}
  The paper~\cite{Beckmann} uses different terminology: \enquote{level} refers
  to what we call the order, while \enquote{degree} refers to the maximum order
  of the types appearing in a term. See also~\cite{AsadaKST19} for a
  probabilistic analysis of $\beta$-reduction length in the simply typed
  $\lambda$-calculus, and~\cite[Section~4.4]{Singh} for something similar in the
  linear $\lambda$-calculus.
\end{rem}

For any simple type $A$, we have $\rmord(A)\leqslant\deg(A)$ and the gap can be arbitrarily large:
\[ \deg(\underbrace{o \to \dots \to o}_{n\ \text{times}} \to o) = n \qquad \rmord(\underbrace{o \to \dots \to o}_{n\ \text{times}} \to o) = 1 \]
While using the degree seems to be the easier way to devise a normalization algorithm together with a proof of a \Tower complexity bound, it seems that the order is the parameter that truly controls complexity. As further evidence towards this (cf.~\cite[\S1.3.6]{titoPhD} for a longer discussion of these three items):
\begin{itemize}
  \item the normalization problem for simply typed $\lambda$-terms of type $\Bool$ containing only subterms whose types have \emph{order at most $2k+2$} is $k$-\ExpTime-complete~\cite{Terui};
  \item in fact, $k$-\ExpTime can be characterized as the predicates expressed by simply typed $\lambda$-terms of \emph{order at most $2k+4$} using a certain input-output convention~\cite{HillebrandKanellakis}, refining an earlier characterization of \Elementary in the simply typed $\lambda$-calculus with no constraint on the order~\cite{HKMairson};
  \item the \enquote{call-by-name translation} ($A \to B =\;!A \multimap B$) from the simply typed $\lambda$-calculus to propositional linear logic (LL) maps the order to the nesting depth of the exponential modality~`$!$', and `$!$' is generally considered to be the main source of complexity in LL (see~\cite{MELLbyLevel} for a connection between the exponential depth in propositional LL and the kind of stratification used in LL-based implicit computational complexity).
\end{itemize}

\section{Safe \texorpdfstring{$\beta$}{beta}-convertibility is Tower-hard}
\label{sec:safe}

As announced in Section~\ref{sec:intro-safe} of the introduction, our goal here is to show:
\begin{thm}\label{thm:safe-hard}
  Given a homogeneous long-safe $\lambda$-term $t$ of type $\Bool$, it is \Tower-hard to decide whether $t =_{\beta} \tttrue$.
\end{thm}
Recall that the type $\Bool$ is defined as $o \to o \to o$. It has exactly two inhabitants in normal form: $\tttrue = \lambda xy.\,x$ and $\ttfalse = \lambda xy.\,y$. They are both safe~\cite[Remark~2.5(ii)]{BlumOng}, and even homogeneous long-safe (as defined in the next subsection).

\subsection{Homogeneity and long-safety}

Before proceeding with the proof of Theorem~\ref{thm:safe-hard}, let us explain the additional restrictions -- beyond the safety condition recalled in Section~\ref{sec:intro-safe} -- imposed on the input terms. We chose to prove hardness in presence of these restrictions because they are natural in the context of higher-order recursion schemes (cf.~\S\ref{sec:hors}).


First, \emph{homogeneity} is a simple syntactic criterion on types.
\begin{defi}
  A simple type $A_1 \to \dots \to A_n \to o$ (for $n\geq0$) is homogeneous when each of the $A_i$ is itself homogeneous ($i\in\{1,\dots,n\}$) and $\rmord(A_1) \geq \dots \geq \rmord(A_n)$.
\end{defi}

The second restriction is \emph{long-safety}. In Church style, it is the same as the notion of safety in \S\ref{sec:intro-safe} except more subterms of a simply typed $\lambda$-term  are considered to be in \enquote{unsafe position}: if $u$ is not a $\lambda$-abstraction, then it is in unsafe position in $t = C[\lambda x.\,u]$. Equivalently, in Curry style:
\begin{defiC}[{\cite[Def.~1.21]{BlumOng}}]
    The typing rules of the long-safe $\lambda$-calculus are
    \[ \frac{}{x : A \vdash x : A} \qquad \frac{\Theta \vdash t : A}{\Theta'
    \vdash t : A}\ (\Theta \subset \Theta') \]
    \[ \frac{\Theta \vdash t : A_1 \to \ldots \to A_n \to B \quad \Theta \vdash
    u_1 : A_1 \;\ldots\; \Theta \vdash u_n : A_n}{\Theta \vdash
    t\,u_1\,\ldots\,u_n : B}\ \text{when}\ \rmord(B) \leqslant
     \inford(\Theta) \]
    \[ \frac{\Theta,\, x_1 : A_1,\, \ldots,\, x_n : A_n \vdash t : B}{\Theta
    \vdash \lambda x_1 \dots x_n.\, t : A_1 \to \dots \to A_n
    \to B}\ \text{when}\ \rmord(A_1 \to \dots \to A_n \to B) \leqslant
    \inford(\Theta) \]
    \[ \text{where}\ \inford(\Theta) = \displaystyle \;\inf_{\mathclap{(y:C)\in\Theta}}\; \rmord(C)\ \text{with the usual convention}\ \inf(\varnothing) = +\infty. \]
\end{defiC}
For the rest of Section~\ref{sec:safe}, we only work with Curry-style typed $\lambda$-terms. This is because it allows our constructions to be carried out in polynomial time. If we had to build Church-style terms, writing out the type annotations would take exponential time, as explained in~\S\ref{sec:pedantic} (though this is harmless in a \Tower-hardness proof; the point is just to have a more precise statement).

\begin{defi}
  For a $\lambda$-term $t$, a simple type $A$ and a context $\Gamma$, we write $\Gamma \hls t : A$ when there is a derivation \emph{using only homogeneous types} of $\Gamma \vdash t : A$ in the above type system. In that case we say that $t$ is \emph{homogeneous long-safe}.
\end{defi}

\subsection{Star-free expressions}

As we said in the introduction, our \Tower-hardness proof proceeds by a reduction from an automata-theoretic problem, which we recall now.

\begin{defi}
  A \emph{star-free expression} over a finite alphabet $\Sigma$ is a regular expression without the Kleene star, but with complementation, defining a language $\denot{E} \subseteq \Sigma^*$:
  \[ E,E' \mathrel{::=} \varnothing \mid \overbrace{\varepsilon}^{\mathclap{\text{empty string}}} \mid \underbrace{a}_{\mathclap{\text{letter}\ a\in\Sigma}} \mid E \cup E' \mid \overbrace{E \cdot E'}^{\mathclap{\text{concatenation}}} \mid \underbrace{\lnot E}_{\mathclap{\text{complement}}} \]
   The \emph{star-free equivalence problem} consists in deciding whether two star-free expressions (over the same alphabet) denote the same language: given $E$ and $E'$, does $\denot{E} = \denot{E'}$ hold?
\end{defi}
For instance, over the alphabet $\Sigma=\{a,b,c\}$, we have the equivalence
\[ \denot{a\cdot\lnot\varnothing \cup b\cdot\lnot\varnothing} = \denot{\lnot(\varepsilon\cup c\cdot\lnot\varnothing)} = \{ w \in \Sigma^* \mid \text{$w$ starts with an $a$ or a $b$} \} \]
Schmitz's paper introducing \Tower presents star-free equivalence as a typical example of a \Tower-complete problem~\cite[\S3.1]{Schmitz} (rephrasing a hardness result proved by Meyer and Stockmeyer, cf.~\cite[\S4.2]{Stockmeyer}). Furthermore, equivalence reduces to emptiness: given two expressions $E$ and $F$, the language denoted by $\lnot(\lnot E \cup F) \cup \lnot(E \cup \lnot F)$ is empty if and only if $E$ and $F$ are equivalent. And to test whether $\llbracket E \rrbracket = \varnothing$, it is well known that it suffices to examine words of length at most $\tower(|E|-1)$:
\begin{prop}[{cf.~e.g.~\cite[proof of Prop.~4.25]{Stockmeyer}}]\label{prop:stockmeyer}
  If a star-free expression of size $n+1$ denotes a non-empty language $L$, then $L$ contains a word of length at most $\tower(n)$.
\end{prop}
(Further remarks on the above proposition and on the complexity of star-free equivalence are given at the end in Section~\ref{sec:dot-depth}.) Altogether, this discussion leads us to conclude that:
\begin{lem}\label{lem:sfempty-hard}
  The following problem is \Tower-hard:
  \begin{itemize}
    \item \emph{Input:} a star-free expression $E$ and a natural number $n$ given in unary.
    \item \emph{Output:} whether or not $\llbracket E \rrbracket$ contains a word of length at most $\tower(n)$.
  \end{itemize}
\end{lem}

\subsection{From star-free expressions to safe \texorpdfstring{$\lambda$}{lambda}-terms}

We shall represent (star-free) languages as functions from strings to booleans. We already have a coding of booleans, that is involved in the statement of Theorem~\ref{thm:safe-hard}. For strings, we use the standard \emph{Church encodings}.

\begin{defiC}[{\cite[\S4]{BohmBerarducci}}]
  The type $\Str_\Sigma = \overbrace{(o \to o) \to \dots \to (o \to
    o)}^{|\Sigma|\ \text{times}} \to o \to o$ is the type of so-called
  \emph{Church encodings} of strings over the finite alphabet $\Sigma$. In the case of a unary alphabet, the type of \emph{Church numerals} is $\Nat = \Str_{\{\mathrm{I}\}} = (o \to o) \to o \to o$.

  The Church encoding of a word $w$ over an ordered alphabet
  $\Sigma=\{c_1,\ldots,c_n\}$ (resp.\ of a number $n\in\naturals$) is
  $\overline{w} = \lambda f_{c_1} \dots f_{c_n} x.\; \underbrace{f_{w[1]}\;(\ldots\;(f_{w[n]}}_{\mathclap{w[i]\ \text{refers to the $i$-th letter of}\ w}}\;x)) : \Str_{\Sigma}$
  (resp.\ $\overline{n} = \overline{\underbrace{\mathrm{I}\dots{}\mathrm{I}}_{\mathclap{n\ \text{times}}}} : \Nat$).
\end{defiC}
For any finite alphabet $\Sigma$ and string $w\in\Sigma^*$, one can check that the type $\Str_\Sigma$ is homogeneous and that $\hls \overline{w} : \Str_\Sigma$ (the safety of $\overline{w}$ is already used in~\cite[\S2.2]{BlumOng}). Next, to encode languages, a fundamental remark is that if $A$ and $B$ are homogeneous simple types and $\hls t : A$, then $\hls t : A[B]$, where:
\begin{nota}
  $A[B]$ denotes the substitution of all occurrences of the base type $o$ in the simple type $A$ by the simple type $B$.
\end{nota}
This allows us to use a $\lambda$-term $t$ such that $\hls t : \Str_\Sigma[B]\to\Bool$ to represent the language $\{ w \in \Sigma^* \mid t\;\overline{w} =_\beta \tttrue \}$ -- this makes sense since $\hls t\;\overline{w} : \Bool$. The key lemma in our proof of Theorem~\ref{thm:safe-hard} is that star-free expressions can be efficiently translated to such representations by (Curry-style) homogeneous safe $\lambda$-terms.

\begin{lem}\label{lem:star-free-lambda}
  A star-free expression $E$ over a finite alphabet $\Sigma$ can be turned in polynomial time into an $\lambda$-term $t_E$ such that
  \begin{itemize}
    \item $\hls t_{E} : \Str_{\Sigma}[A_{E}] \to \Bool$ for some simple type $A_E$;
    \item for any $w\in\Sigma^*$, we have $t_{E}\, \overline{w} =_\beta \tttrue$ if and only if $w \in \llbracket E \rrbracket$.
  \end{itemize}
(The type $A_E$ is not part of the output of the polynomial-time algorithm taking $E$ and $\Sigma$ as input. In fact, the size of $A_E$ may be exponential in the size of $E$.)
\end{lem}

To prove Lemma~\ref{lem:star-free-lambda}, we first note that boolean operations can be implemented by 
\[ \quad \ttand = \lambda b_1 b_2 xy.\; b_1\; (b_2\; x\; y)\; y \qquad \ttor = \lambda b_1 b_2 xy.\; b_1\; x\; (b_2\; x\; y) \qquad \ttnot = \lambda bxy.\; b\;y\;x \]
These $\lambda$-terms satisfy $\hls \ttand,\ttor : \Bool \to \Bool \to \Bool$ and $\hls \ttnot : \Bool \to \Bool$. Furthermore, we can implement concatenation on strings (cf.~\cite[Theorem~2.8]{BlumOng}) with
\[ \ttconc = \lambda s_1s_2 f_1 \dots f_{|\Sigma|} x.\; s_1\; f_1\; \dots\; f_{|\Sigma|}\; (s_2\; f_1\; \dots\; f_{|\Sigma|}\; x) \]
which is a homogeneous long-safe $\lambda$-term of type $\Str_\Sigma\to\Str_\Sigma\to\Str_\Sigma$. This generalizes to $k$-fold concatenation: $\hls \ttconc_{k} : \underbrace{\Str_\Sigma \to \dots \to \Str_\Sigma}_{k\ \text{times}} \to \Str_\Sigma$.

The next step in our proof is the following lemma, which is reused in Section~\ref{sec:conclude-proof}.
\begin{lem}\label{lem:test-all-strings}
  For every finite alphabet $\Sigma$, one can build in time $|\Sigma|^{O(1)}$ a $\lambda$-term $\ttany_\Sigma$ such that, for any homogeneous type $A$, there is some $F_{\ttany}(A)$ such that
  \[ \hls \ttany_\Sigma : \Str_{\Sigma\cup\{\#\}}[F_{\ttany}(A)] \to (\Str_{\Sigma}[A] \to \Bool) \to \Bool\]
  and for every list of words $w_0,\dots,w_n \in \Sigma^*$ and every $\lambda$-term $t : \Str_{\Sigma}[A] \to \Bool$,
  \[ \ttany_\Sigma\; \overline{w_0\# \dots \#w_n}\; t =_\beta \tttrue \quad\iff\quad \exists i \in \{0,\dots,n\}.\; t\, \overline{w_i} =_\beta \tttrue \]
\end{lem}
\noindent
\begin{proof}
   Let $\Sigma = \{c_{1},\dots,c_{|\Sigma|}\}$. Take $F_{\ttany}(A) = \Str_{\Sigma}[A] \to \Bool$ and define $\ttany_\Sigma$ as
  \[ \lambda sp.\; s\; \overbrace{u_{1}\; \dots\; u_{|\Sigma|}}^{\mathclap{\text{correspond to letters in}\ \Sigma}}\; (\underbrace{\lambda fx.\; \ttor\; (p\; x)\; (f\;\emptycode)}_{\mathclap{\text{corresponds to}\ \#}})\; p\; \emptycode \quad \text{where}\
    u_{i} = \lambda fx.\; f\; (\ttconc\; x\; \overline{c_i}) \]
  with $x, \emptycode : \Str_{\Sigma}[A]$, $p,f : \Str_{\Sigma}[A] \to \Bool$, $u_{i} : (\Str_{\Sigma}[A] \to \Bool) \to \Str_{\Sigma}[A] \to \Bool$ and $s : \Str_{\Sigma\cup\{\#\}}[\Str_{\Sigma}[A] \to \Bool]$ in the typing derivation for $\ttany_{\Sigma}$.

  Explanation: computing $\ttany_\Sigma\; \overline{w_0\# \dots \#w_n}\;t$ performs a \enquote{right fold} over the input with an accumulator of type $\Str_{\Sigma}[A] \to \Bool$, representing a language $L(\sigma) \subseteq \Sigma^*$ where $\sigma$ is the suffix processed up until this point of the computation. The idea is that
  \[ u \in L(v\#w_m\#\dots\#w_n) \quad\iff\quad t\;\overline{uv} =_\beta \tttrue \lor \exists i \geq m.\; t\;\overline{w_i} =_\beta \tttrue \]
  This makes the condition that we want to check equivalent to $\varepsilon\in L(w_0\#\dots\#w_n)$.
\end{proof}

\begin{proof}[Proof of Lemma~\ref{lem:star-free-lambda}]
  By induction on the expression, leading to a recursive algorithm. We also want the algorithm to return $\rmord(A_E)$ on the input $E$, because this information on subexpressions is used in two cases to make sure that the term we build admits a homogeneous type. It will be clear that in every case, $\rmord(A_E)$ can be recursively computed and is bounded by $O(|E|)$.
  \begin{itemize}
  \item We take $A_\varnothing = o$ and $t_\varnothing = \lambda s.\; \ttfalse :
    \Str_\Sigma \to \Bool$.
  \item In the simply typed $\lambda$-calculus, testing for the empty word can be
    done with type $\Str_\Sigma \to \Bool$, but this is not possible with safe $\lambda$-terms
    as discussed in~\cite[Section~2]{BlumOng}. We use instead $A_\varepsilon =
    \Bool$ and
    \[ t_\varepsilon = \lambda s.\; s\; (\lambda x.\; \ttfalse)\; \dots\;
      (\lambda x.\; \ttfalse)\; \tttrue \]
  \item To test whether the word consists of a single letter, say, the first one in
    the alphabet $\Sigma$ (call it $c_1$), we use $A_{c_1} = o \to o \to o \to
    o$ and define $t_{c_1}$ to be a $\lambda$-term corresponding to a
    deterministic finite automaton that recognizes the language $\{c_1\}$, namely:
    \[ \lambda sxy.\; s\; (\overbrace{\lambda qz_0z_1z_2.\; q\; z_1\; z_2\; z_2}^{\mathclap{\qquad\qquad\text{represents the transition}\ \delta(c_1,q_0)=q_1,\; \delta(c_1,q_1)=\delta(c_1,q_2)=q_2}})\; \underbrace{(\lambda q.\; q_2)\; \dots\;(\lambda q.\;q_2)}_{|\Sigma|-1\ \text{times}}\;q_0\; \overbrace{y\;x\;y}^{\mathclap{\qquad\text{accepting states}\;=\;\{q_1\}}} \quad\text{where}\ \underbrace{q_i = \lambda z_0z_1z_2.\; z_i}_{\mathclap{\text{the 3 states of the DFA (of type $A_{c_1}$)}\qquad\qquad}} \]
  \item Complementation is trivially implemented by post-composing with $\ttnot$.
  \item To handle a union $E\cup F$, we first check whether $\rmord(A_E)\geq\rmord(A_F)$. If that is the case, we take $A_{E\cup F} = \Str_\Sigma[A_E] \to \Str_\Sigma[A_F] \to
    \Bool$ and
    \[ t_{E\cup F} = \lambda s.\; s\; \underbrace{\dots\;(\lambda kxy.\; k\; (\ttconc\;x\;\overline{c_i})\; (\ttconc\;y\;\overline{c_i}))\; \dots}_{\text{for each $c_i\in\Sigma$}}\; (\lambda xy.\; \mathtt{or}\; (t_E\;x)\; (t_F\;y))\; \emptycode\; \emptycode \]
    with $k : \Str_\Sigma[A_E] \to \Str_\Sigma[A_F] \to \Bool$, $x :
    \Str_\Sigma[A_E]$ and $y : \Str_\Sigma[A_F]$ in the typing derivation; the
    two occurrences of $\emptycode$ at the end must be given
    the different types $\Str_\Sigma[A_E]$ and $\Str_\Sigma[A_F]$ respectively.

    This can be seen as a continuation-passing-style transformation of the following procedure using pair types (which we do not have here): perform a \enquote{left fold} over the input string with an accumulator of type $\Str_\Sigma[A_E]\times\Str_\Sigma[A_F]$; for each input letter $c$, apply the function $(w,w')\in\Sigma^*\mapsto (wc,w'c)$, so that the result of the fold is two copies of the input string with different types; finally, check on $(w,w')$ whether $w\in\denot{E}$ or $w'\in\denot{F}$.
    
    When $\rmord(A_E)<\rmord(A_F)$, we take $t_{E\cup F} = t_{F\cup E}$. Observe
    that $t_E$ and $t_F$ appear only once in $t_{E\cup F}$ -- this is true in
    every recursive case, ensuring that $t_E$ has size $O(|E|)$. 
  \item The remaining case, concatenation, is the most delicate.
  
  First, we introduce a new letter $\square\notin\Sigma$ and build a
  $\lambda$-term $t'_{E,F}$ that computes the language $\denot{(E\square)^*F}$. When $\rmord(A_E)\geq\rmord(A_F)$, we get
      \[ \hls t'_{E,F} : \Str_{\Sigma\cup\{\square\}}[\Str_\Sigma[A_E] \to \Str_\Sigma[A_F] \to \Bool] \to \Bool \]
      by taking $t'_{E,F} = \lambda s.\; s\; v_{c_1}\; \dots\; v_{c_{|\Sigma|}}\; v_{\square}\; u\; \emptycode\; \emptycode$ (for $\Sigma=\{c_1,\dots,c_{|\Sigma|}\}$) where
      \[ \qquad v_c = \lambda fxy.\; f\; (\ttconc\;x\;\overline{c})\; (\ttconc\;y\;\overline{c})  \qquad  v_\square = \lambda fxy.\; \ttand\;(t_E\;x)\;(f\;\emptycode\;\emptycode) \qquad u = \lambda xy.\; t_F\; y\]
      Similarly to the proof of Lemma~\ref{lem:test-all-strings}, $t'_{E,F}$
      performs a \enquote{right fold} on its input string (over the alphabet
      $\Sigma\cup\{\square\}$), whose accumulator $f : \Str_\Sigma[A_E] \to
      \Str_\Sigma[A_F] \to \Bool$ represents a language (over $\Sigma$) --
      except that, analogously to the case $E \cup F$ above, this representation
      must be fed two copies $x : \Str_\Sigma[A_E]$ and $y :
      \Str_\Sigma[A_F]$ of the same string. We get
      \[ \{w \in \Sigma^* \mid w\sigma \in \denot{(E\square)^*F} \} \]
      after reading the input suffix $\sigma\in(\Sigma\cup\{\square\})^*$; therefore, once the whole input string has been traversed, we just need to
      check that the final accumulator contains $\varepsilon$.

      If $\rmord(A_E) < \rmord(A_F)$, to ensure homogeneity, we replace in the definition of $t'_{E,F}$
      \[ v_\square\ \text{by}\ \lambda kxy.\; \ttand\;(t_E\;y)\;(k\;\emptycode\;\emptycode) \quad\text{and}\quad u\ \text{by}\ \lambda xy.\; t_F\; x \]
      leading to $\hls t'_{E,F} : \Str_{\Sigma\cup\{\square\}}[\Str_\Sigma[A_F]
      \to \Str_\Sigma[A_E] \to \Bool] \to \Bool$.

      Then, the $\lambda$-term that we want is $t_{E\cdot F} = \lambda s.\;
      \overbrace{\ttany_{\Sigma\cup\{\square\}}}^{\mathclap{\text{Lemma~\ref{lem:test-all-strings}
            above}\qquad\quad}}\;
      (\overbrace{\ttsplit_\Sigma}^{\mathclap{\quad\qquad\text{Lemma~\ref{lem:split}
            below}}}\; s)\; t'_{E,F}$ where the role of $t'_{E,F}$ is to
      distinguish, among the strings containing exactly one `$\square$', those
      that belong to $\denot{E\square F}$ -- indeed,
      $\denot{(E\square)^*F}\cap\Sigma^*\square\Sigma^* = \denot{E\square F}$.
      \qedhere
  \end{itemize}
\end{proof}
\noindent
We isolate the following part of the above proof to serve as an example of independent interest (cf.~\S\ref{sec:iatlc}) of a string-to-string function expressed in the safe $\lambda$-calculus.
\begin{lem}\label{lem:split}
  For every finite alphabet $\Sigma$, one can build in time $|\Sigma|^{O(1)}$ a $\lambda$-term $\ttsplit_\Sigma$ such that, for some type $A_\ttsplit$,
  \[ \hls \ttsplit_\Sigma : \Str_{\Sigma}[A_{\ttsplit}] \to \Str_{\Gamma} \quad\text{where}\ \Gamma = \Sigma\cup\{\square,\#\} \]
  and for every list of letters $a_1,\dots,a_n \in \Sigma$,
  \[ \qquad\ttsplit_\Sigma\; \overline{a_1\dots a_n}\quad=_\beta\quad \overline{\square a_1 \dots a_n \# a_1\square a_2 \dots a_n \# \dots \# a_1 \dots a_{n-1}\square a_n \# a_1 \dots a_n \square}  \]
\end{lem}
\begin{proof}
  As in the case of unions in the proof of Lemma~\ref{lem:star-free-lambda}, the
  idea is to start from a left-to-right procedure using a pair type, and then
  apply a continuation-passing-style transformation. For $p =
  a_1 \dots a_m \in\Sigma^*$ (where $a_i\in\Sigma$), let us consider the maps
  \begin{align*}
    X_p \colon \naturals &\to \Sigma^* & F_{a_1 \dots a_m} \colon \Gamma^* &\to \Gamma^*\\
    n &\mapsto p & y &\mapsto \square a_1 \dots a_m y a_1\square a_2 \dots a_m y  \dots y a_1 \dots a_{m-1} \square a_m y
  \end{align*}
  In particular, $F_\varepsilon(y) = \varepsilon$ and $F_a(y) = \square a y$ for
  $a\in\Sigma$. The maps $X_p$ are constant; the role of their dummy argument is to
  increase the order of an associated variable in the $\lambda$-term defined
  below so that it satisfies the safety condition --
  $\rmord(\Nat\to\Str_\Gamma)=3>2=\rmord(\Str_\Gamma)$.
  
  For $p\in\Sigma^*$ and $c\in\Sigma$, we have the inductive equations
  \[ X_{pc}(n) = X_p(0) \cdot c \qquad F_{pc}(y) = F_p(cy) \cdot X_p(0) \cdot
    \square c \cdot y \]
  which we translate into the $\lambda$-term
  \[ v_c = \lambda kxf.\; k\; (\lambda n.\; \ttconc\; (x\;\overline{0})\; \overline{c})\;
    (\underbrace{\lambda y.\; \ttconc_4\; (f\; (\ttconc\; \overline{c}\; y))\;
      (x\;\overline{0})\;    \overline{\square{}c}\; y}_{\mathclap{\text{$x$
          appears free in this subterm of order 3 $\rightsquigarrow$ the
          $\Nat\to\Str_\Gamma$ trick ensures (long-)safety}\qquad}})
    \quad\text{for}\ c\in\Sigma \]
  with $k : (\Nat\to\Str_\Gamma) \to
  (\Str_\Gamma\to\Str_\Gamma)\to\Str_\Gamma$, $x : \Nat\to\Str_\Gamma$, $f :
  \Str_\Gamma\to\Str_\Gamma$, $n : \Nat$ and $y : \Str_\Gamma$.
  Finally, we take
  \[\ttsplit = \lambda s.\; s\; v_{c_1}\; \dots\;
  v_{c_{|\Sigma|}}\; (\lambda xf.\; \ttconc_3\; (f\;\overline\#)\;
  (x\;\overline{0}) \;\overline\square)\; (\lambda n.\;\emptycode)\; (\lambda
  y.\; \emptycode)\] -- thus, $A_\ttsplit$ is the aforementioned type of $k$.
\end{proof}

\subsection{Proof of Theorem~\ref{thm:safe-hard}}
\label{sec:conclude-proof}

We now have all the ingredients to reduce the \Tower-hard decision problem of
Lemma~\ref{lem:sfempty-hard} to safe $\beta$-convertibility. This reduction is
performed by the following lemma, which therefore immediately entails our main
theorem.
\begin{lem}\label{lem:enum}
  Given a star-free expression $E$ over a finite alphabet $\Sigma$, and $n\in\naturals$ written in unary, one can build in polynomial time a $\lambda$-term $b_E$ such that $\hls b_E : \Bool$ and
  \[ b_E =_\beta \tttrue \quad\iff\quad \denot{E} \cap \Sigma^{\leqslant\tower(n)} \neq \varnothing \]
\end{lem}
\begin{proof}
  The key is to build a $\lambda$-term $\ttenum_\Sigma$ such that:
  \begin{itemize}
    \item $\hls \ttenum_\Sigma : \Nat[A_{\ttenum}] \to \Str_{\Sigma\cup\{\#\}}$ for some type $A_{\ttenum}$;
    \item for any $m\in\naturals$, the application $\ttenum\;\overline{m}$ reduces to the Church encoding of a $\#$-separated list containing all words in $\Sigma^{*}$ of length at most $m$.
  \end{itemize}
  We define it for $\Sigma = \{c_{1},\dots,c_{|\Sigma|}\}$ in such a way that $\ttenum_\Sigma\;\overline{N} =_\beta \overline{f_N(\#)}$ for $N\in\naturals$, where the sequence of functions $f_n\colon (\Sigma\cup\{\#\})^* \to (\Sigma\cup\{\#\})^*$ is inductively defined by
  \[ f_0(x) = x \qquad f_{n+1}(x) = x \cdot f_n(c_1 x) \cdot \ldots \cdot f_n(c_{|\Sigma|} x) \]
  e.g.\ for $\Sigma=\{\mathtt{0,1}\}$ we have $f_1(x) = x\mathtt{0}x\mathtt{1}x$ and $f_2(x) = x \mathtt{0}x\mathtt{00}x\mathtt{10}x \mathtt{1}x\mathtt{01}x\mathtt{11}x$. Observe that the string $f_2(\#)$ is a $\#$-separated representation of the list $[\varepsilon,\mathtt{0,00,10,1,01,11},\varepsilon]$ that indeed contains all words of length at most 2 over $\Sigma$ (with one redundancy).
  To implement this, take $A_{\ttenum} = \Str_{\Sigma} \to \Str_{\Sigma}$ and
  \[ \ttenum_\Sigma = \lambda n.\; n\; (\lambda fs.\;
    \ttconc_{|\Sigma|+1}\; s\; (f\;(\ttconc\;
    \overline{c_{1}}\; s))\; \dots \; (f\; (\ttconc\;
    \overline{c_{|\Sigma|}}\; s)))\; (\lambda y.\; y)\; \overline{\#} \]
  Once we have that, we can take $b_E = \overbrace{\ttany_\Sigma}^{\mathclap{\text{Lemma~\ref{lem:test-all-strings}}}}\;(\ttenum_\Sigma\; \tttow_n)\overbrace{t_E}^{\mathclap{\text{Lemma~\ref{lem:star-free-lambda}}}}$ where
      $\tttow_n = \overbrace{\overline{2}\; \dots\; \overline{2}}^{n\ \text{times}}$ satisfies $\hls \tttow_n : \Nat$ and is $\beta$-convertible to $\overline{\tower(n)}$ (cf.~\cite[Remark~3.5(iii)]{BlumOng}).
\end{proof}

\subsection{Final remarks}
To conclude this paper, we discuss some topics related to Theorem~\ref{thm:safe-hard} and its proof.

\subsubsection{On the complexity of star-free expression equivalence}
\label{sec:dot-depth}

The equivalence problem for usual regular expressions, with Kleene star but no complementation, is \enquote{merely} \PSpace-complete~\cite[Theorem~4.13]{Stockmeyer} -- see also~\cite[Theorem~15]{ComplexityRegularLike}. One could therefore think that the main source of complexity in the case of star-free expressions is \emph{complementation}.

This can be witnessed in the proof idea for Proposition~\ref{prop:stockmeyer}. Translate $E$ to an equivalent nondeterministic finite automaton (NFA), whose number of states bounds the length of a shortest word in the language (since the shortest accepting runs visit each state at most once). This can be done by induction on $E$, using any standard construction on NFA for union and concatenation; the costliest operation is complementation, which uses determinization, inducing a single exponential state blowup.

But in our translation from star-free expressions to $\lambda$-terms (Lemma~\ref{lem:star-free-lambda}), it turns out that complementation is trivial while concatenation increases the order of the types. Relatedly, in \emph{deterministic} finite automata, complementation can be done without changing the set of states, while concatenation may need an exponential blowup of the number of states. So the complexity of the problem is in fact rooted in the \emph{alternation} between complementation and concatenation -- which leads to the \enquote{dot-depth hierarchy}~\cite{DotDepth}.

\subsubsection{Safety in higher-order recursion schemes}
\label{sec:hors}

Blum and Ong's main motivation for introducing the safe $\lambda$-calculus~\cite{BlumOng} was unrelated to complexity: they wanted to transpose to the $\lambda$-calculus the notion of safety on higher-order recursion schemes (HORS)~\cite{SafeHORS} -- and indeed, the correspondence between HORS and $\lambda\mathbf{Y}$-terms (i.e.\ terms in the simply typed \mbox{$\lambda$-calculus} extended with fixed-point combinators $\mathbf{Y}_A : (A \to A) \to A$) relates safe HORS with safe $\lambda\mathbf{Y}$-terms~\cite{SalvatiWalukiewicz}. Even without $\mathbf{Y}$, the safe $\lambda$-calculus has some interesting properties: for example, when substituting a safe $\lambda$-term into another, no spurious variable capture can happen~\cite[Lemma~1.10]{BlumOng} -- this leads to a way to normalize long-safe $\lambda$-terms by rewriting without $\alpha$-renaming, but it is a bit subtle, see~\cite[Section~5.2]{Avoidance}. A universal algebra perspective on the safe $\lambda$-calculus is also sketched in~\cite[\S2.3]{SalvatiHDR}.

The homogeneity condition also comes from the study of
HORS, cf.~\cite{ParysHomogeneity}. In fact, the homogeneous long-safe $\lambda$-calculus is equivalent (cf.~\cite[Remark~3.53]{BlumPhD}) to an earlier attempt~\cite[Section~2.4.2]{MirandaHOMC} to design a safe $\lambda$-calculus inspired by HORS.

\subsubsection{Implicit automata in typed $\lambda$-calculi and transducer theory}
\label{sec:iatlc}

Lemma~\ref{lem:star-free-lambda} tells us that every language of the form $\denot{E}$ for some star-free expression can be expressed by some homogeneous long-safe $\lambda$-term of type $\Str_\Sigma[A] \to \Bool$. These \emph{star-free languages} are a fundamental and well-studied (cf.~\cite{StraubingSIGLOG}) subclass of \emph{regular languages}. Hillebrand and Kanellakis have shown that the languages computed by simply typed $\lambda$-terms of type $\Str_\Sigma[A] \to \Bool$ are, in fact, exactly the regular languages~\cite[Theorem~3.4]{HillebrandKanellakis}\footnote{This paper~\cite{HillebrandKanellakis} is the same that was already mentioned in \S\ref{sec:stlc-order}.}; they have a translation of deterministic finite automata (rather than regular expressions) into the simply typed $\lambda$-calculus that actually produces safe $\lambda$-terms.
In fact, we used this translation of DFA in the proof of Lemma~\ref{lem:star-free-lambda}, for the \enquote{single letter} case of the induction.

By using an affine and non-commutative type system, one can get a variant of Hillebrand and Kanellakis's theorem characterizing star-free languages instead~\cite{aperiodic}.\footnote{See also~\cite[Chapter~7]{titoPhD} for a variant of this result (with linear instead of affine types).} The proof does not translate star-free expressions; instead, to encode star-free languages, it goes through an algebraic characterization. Nevertheless, these works were the main inspirations for our proof strategy for Theorem~\ref{thm:safe-hard}.

Further works~\cite{titoPhD,pradic2024implicit} in this \enquote{implicit automata in typed $\lambda$-calculi} research programme have focused on characterizations of string-to-string (or tree-to-tree) functions computed by \emph{transducers}, i.e.\ automata with output. For reasons that are close to the aforementioned connection with HORS, as discussed in~\cite[\S1.4.1]{titoPhD}, definability in the safe $\lambda$-calculus characterizes an important class of functions from 1980s transducer theory. The functions defined by the safe $\lambda$-terms $\ttsplit_{\Sigma}$ and $\ttenum_{\Sigma}$ (from Lemmas~\ref{lem:split} and~\ref{lem:enum} respectively) therefore belong to this class. The \enquote{split} function is also a typical example of the more recently introduced \emph{polyregular functions}, cf.~\cite{PolyregSurvey,Kiefer24} -- in fact, it is a slight variation on the \enquote{squaring with underlining} function of~\cite[Example~3]{PolyregSurvey}.

\bibliographystyle{alphaurl}
\bibliography{safe-complexity}

\end{document}